\documentclass[letterpaper, 10 pt, conference]{ieeeconf}  
\usepackage{gen_settings}

\IEEEoverridecommandlockouts                              

\title{\LARGE \bf
Disturbance Bounds for Signal Temporal Logic Task Satisfaction: \\
A Dynamics Perspective
}

\author{Prithvi Akella, Aaron D. Ames$^{1}$
\thanks{*This work was supported by AFOSR}
\thanks{$^{1}$All authors are with the California Institute of Technology
        {\tt\small pakella@caltech.edu, ames@caltech.edu}}%
}

\begin{document}

\maketitle
\thispagestyle{empty}
\pagestyle{empty}

\begin{abstract}
This letter offers a novel approach to Test and Evaluation of pre-existing controllers from a control barrier function and dynamics perspective.  More aptly, prior Test and Evaluation techniques tend to require \textit{apriori} knowledge of a space of allowable disturbances.  Our work, however, determines a two-norm disturbance-bound rejectable by a system's controller \textit{without requiring specific knowledge of these disturbances beforehand}.  The authors posit that determination of such a disturbance bound offers a better understanding of the robustness with which a given controller achieves a specified task - as motivated through a simple, linear-system example.  Additionally, we show that our resulting disturbance bound is accurate through simulation of $1000$ randomized trials in which a Segway-controller pair successfully satisfies its specification despite randomized perturbations within our identified bound.

\end{abstract}

\section{INTRODUCTION}
While there exist multiple temporal logic formalisms, two of increasing interest in the controls community are Linear Temporal Logic and Signal Temporal Logic~\cite{baader2003description,baier2008principles,donze2010robust}.  This interest arises as these logical schemes offer succinct ways of expressing complex, desired behavior, while also providing necessary and sufficient criteria by which to determine if a system has achieved this behavior~\cite{vardi1996automata, gerth1995simple, donze2010robust, maler2004monitoring}.  As a result, there has been significant work utilizing these specification schemes and associated satisfaction criteria to develop optimization-based control schemes that enforce satisfaction of these behavioral specifications~\cite{wolff2014optimization,wongpiromsarn2011tulip,lindemann2018control,lindemann2019control,lindemann2020barrier}.  Additionally, these formalisms and satisfaction criteria have also prompted the development of evaluation schemes to test a controllers ability to realize these desired system behaviors when experiencing environmental disturbances~\cite{Althoff2018,annpureddy2011s,tuncali2016utilizing,fainekos2012verification,dreossi2019verifai,gangopadhyay2019identification,ghosh2018verifying}.  Finally, the authors note that there has also been significant work aimed at developing controllers that robustly reject these environmental disturbances, most recently with active disturbance rejection control~\cite{han2009pid,huang2014active,gao2006active,sun2007comments}.

However, this leads to a question we aim to explore in this work.  As mentioned prior, existing work in the Test and Evaluation community endeavors to test and evaluate a controller's ability to realize desired system specifications while subject to environmental disturbances.  These procedures typically amount to an optimization problem over the feasible space of these disturbances, requiring identification of the allowable disturbances beforehand~\cite{corso2020survey}.  As such, the authors posit that it might be more fruitful were we to identify the level of disturbance that a given controller can reject as opposed to determining the worst-case disturbance from a given set.  More accurately, \textit{can we use a system model and model-theoretic control techniques to identify a two-norm disturbance-bound that our controller can reject whilst still satisfying its incumbent specification?}

\spacing
\newidea{Our Contribution:}  Our contribution is twofold.  First, we construct two optimization problems that each generate two-norm disturbance-bounds rejectable by a system's controller while it steers its system to satisfy its specification.  Each optimization problem focuses on a specific subset of Signal Temporal Logic, and we use their solutions to construct our system-level bound.  Secondly, we show that our generated bound is accurate.  Over $1000$ simulated Segway runs where disturbances are sampled randomly from within our prescribed norm-bound, we show that the Segway-controller pair rejects disturbances within our identified bound and achieves its Signal Temporal Logic task. For context, the subset of STL tasks studied in the sequel is consistent with prior works in the controls literature~\cite{lindemann2018control,lindemann2019control,lindemann2019decentralized}.

\spacing
\newidea{Organization:}  Section~\ref{sec:probform} details some background material in Subsection~\ref{sec:prelims}, motivates our problem in Subsection~\ref{sec:motiv_example}, and formally states our problem in Subsection~\ref{sec:probstate}. Then, Section~\ref{sec:contributions} details our main contributions - the optimization problems determining two-norm disturbance-bounds rejectable by a system's controller.  Finally, Section~\ref{sec:examples} illustrates our results through a simulated Segway example.

\section{Problem Formulation}
\label{sec:probform}
This section will detail some necessary background material for the sequel - specifically Signal Temporal Logic and Control Barrier Functions.  We will start with some notation.

\spacing
\newidea{Notation:} $\|\cdot\|$ is the 2-norm over $\mathbb{R}^n$.  $\mathbb{R}_{+} = \{x \in \mathbb{R}~|~x \geq 0\}$, $\mathbb{R}_{++} = \{x \in \mathbb{R}~|~x > 0\}$.  A function $f: \mathbb{R}^n \to \mathbb{R}$ is Lipschitz continuous if and only if $\exists~L \in \mathbb{R}_{+}$ such that $|f(x) - f(z)| \leq L \|x -z \|$.  A continuous function $\alpha \in \classkapextinf$ if and only if $\alpha: (-\infty,\infty) \to \mathbb{R}$, $\alpha(0) = 0$, $r>s$ implies $\alpha(r) > \alpha(s)$, and $\lim_{r \to \infty}~\alpha(r) = \infty$.  For any continuously differentiable function $h:\mathbb{R}^n \to \mathbb{R},$ $a \in \mathbb{R}$ is a regular value if and only if $D_xh(x) \neq 0~\forall~x \suchthat h(x) = a$.  The space of all signals $\signalspace = \{s~|~s:[0,T] \to \mathbb{R}^n,~\forall~T > 0\}$ with $s$ a signal. $\|\cdot\|_{[a,b]}$ is an induced (semi)-norm over $\signalspace$ where $\|s\|_{[a,b]} = \max_{t \in [a,b]}\|s(t)\|$ for $s \in \signalspace$.
\subsection{Preliminaries}
\label{sec:prelims}
In this section, we will provide a brief description of Signal Temporal Logic and Control Barrier functions - two topics that are necessary for the sequel.  Afterwards, we will motivate the specific problem under study with an example.

\spacing
\newidea{Signal Temporal Logic:} Signal Temporal Logic (STL) is a language by which rich, time-varying system behavior can be succinctly expressed.  This language is based on predicates $\mu \in \mathcal{A}$ which are boolean-valued variables taking a truth value for each state $x$.  Predicates $\mu$ and specifications $\psi$ are defined as follows, with $"|"$ demarcating definitions:
\begin{gather}
    \label{eq:predicate_def}
    \mu(x) = \true \iff h_{\mu}(x)\geq 0,~h_{\mu}: X \to \mathbb{R}, \\
    \label{eq:spec}
    \psi \triangleq \phi | \neg \psi | \psi_1 \lor \psi_2 | \psi_1 \wedge \psi_2 | \psi_1 \until_{[a,b]} \psi_2,~ \psi \in \mathbb{S}.
\end{gather}
Here, $\psi_1,\psi_2$ are specifications themselves, and $\psi_1 \until_{[a,b]} \psi_2$ reads as: $\psi_1$ should be true at time $t = a$ and should continue to be true until $\psi_2$ is true, which should be true by some time $t \leq b$ \cite{baier2008principles,donze2010robust}.
Finally, $\mathbb{S}$ is the set of all STL specifications.

We write $(s,t')\models \psi$ when a signal $s$ satisfies a specification $\psi$ for times $t \geq t'$.  To be brief, will refrain from formally defining the satisfaction relation $\models$, as we will instead note that every STL specification $\psi$ has a robustness measure $\rho$ that is positive for signals $s$ that satisfy $\psi$.
\begin{definition}
\label{def:robustness}
A function $\rho: \signalspace \times \mathbb{R}_+ \to \mathbb{R}$ is a \textit{robustness measure} for a Signal Temporal Logic specification $\psi$ if it satisfies the following equivalency:
\begin{equation}
    \rho(s,t) \geq 0 \iff (s,t) \models \psi.    
\end{equation}
\end{definition}
\noindent For a more comprehensive definition of the satisfaction relation, please see Section 2.2 in~\cite{maler2004monitoring}.  Finally, to simplify notation, two commonly used temporal logic operators will be produced here.  The first is $\F_{[a,b]} \psi$ which reads as $\psi$ should be true \textit{at some point in the future} for some time $t \in [a,b]$.  The second is $\G_{[a,b]} \psi$ which reads as $\psi$ should be true \textit{for all times} $t \in [a,b]$.  In both cases, $b>a$.
\begin{equation}
    \F_{[a,b]}\psi = \true \until_{[a,b]} \psi, ~ \G_{[a,b]}\psi = \neg\left(\true \until_{[a,b]} \neg \psi\right).
\end{equation}

\spacing
\newidea{Control Barrier Functions:} Originally inspired by their counterparts in optimization (see Chapter 3 of \cite{forsgren2002interior}), control barrier functions are a modern control tool used to ensure safety in safety-critical systems that are control-affine, \textit{i.e.},
\begin{equation}
    \label{eq:base_system}
    \dot x = f(x) + g(x) u, \quad x \in \mathcal{X} \subseteq, \mathbb{R}^n,~u\in\mathcal{U}\subseteq \mathbb{R}^m.
\end{equation}
We will assume we have a feedback controller $k(x)$ for~\eqref{eq:base_system}, which results in the following closed-loop dynamics:
\begin{equation}
    \label{eq:nom_sys}
    \dot{x} = f_{cl}(x) \triangleq f(x) + g(x)k(x),~x \in \mathcal{X}.
\end{equation}
Now, solutions to~\eqref{eq:nom_sys} may not exist for all time~\cite{verhulst2006nonlinear}.  As such, we denote this interval of existence of solutions to~\eqref{eq:nom_sys} emanating from $x_0$ as $I(x_0)=[0,t_{\max}]$.  We denote the corresponding solution as $\phi_t(x_0)$, where
\begin{equation}
    \label{eq:solution}
    \dot{\phi}_t(x_0) = f_{cl}\left(\phi_t(x_0) \right),\quad \phi_0(x_0) = x_0.
\end{equation}
Then, forward invariance is defined as follows.
\begin{definition}
The set $\mathcal{C} \subset \mathbb{R}^n$ is forward invariant with respect to the dynamical system~\eqref{eq:nom_sys} if $\forall~x_0 \in \mathcal{C}$, $\phi_t(x_0) \in \mathcal{C}~\forall~t \in I(x_0)$, with $\phi_t(x_0)$ as per~\eqref{eq:solution}.
\end{definition}

\begin{figure}
    \centering
    \hspace{-0.15 cm}\includegraphics[width = 0.49\textwidth]{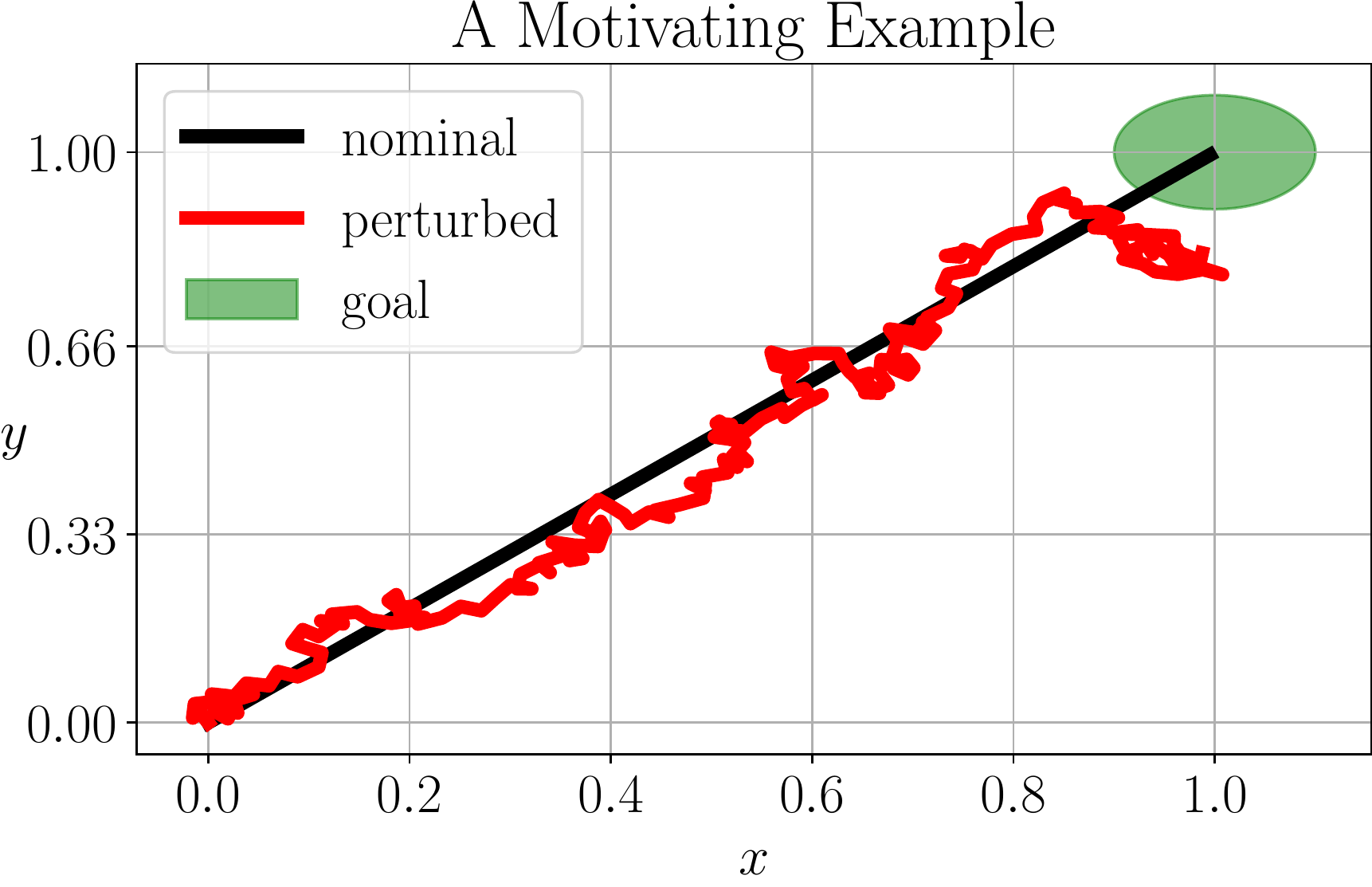}
    \vspace{-0.7 cm}
    \caption{The motivating example detailed in Section~\ref{sec:motiv_example} for this paper's problem.  For the closed-loop system shown, the undisturbed trajectory (black) satisfies its specification - reach the goal (green) within $2$ seconds.  However, disturbing the same system results in a trajectory (red) that fails to satisfy this specification.  This phenomenon prompted the authors to ask the question - \textit{can we determine the two-norm disturbance-bound that a controller can reject while steering a system to satisfy its specification?}}
    \vspace{-0.5 cm}
    \label{fig:example_trajectory}
\end{figure}

Control barrier functions then, are a tool used to ensure forward invariance of their $0$-superlevel sets.  Specifically, for a continuously differentiable function $h: \mathbb{R}^n \to \mathbb{R}$, define its $0$-superlevel set $\mathcal{C}$ and boundary $\partial \mathcal{C}$  as follows:
\begin{equation}
    \label{eq:zerosuplevel}
    \mathcal{C}= \{x\in\mathcal{X}~|~h(x) \geq 0\},~\partial \mathcal{C} = \{x\in\mathcal{X}~|~h(x) = 0\}.
\end{equation}
Then, the definition of control barrier functions is as follows.
\begin{definition} (Adapted from Definition 5 in~\cite{ames2016control})
\label{def:continuous_cbf}
For the control-affine system~\eqref{eq:base_system}, a continuously differentiable function $h:\mathbb{R}^n \to \mathbb{R}$ with $0$ a regular value is a \textit{control barrier function} if $\exists~\alpha\in\classkapextinf$ such that $\forall~x \in \mathcal{X}$,
\begin{equation}
    \label{eq:cbf_criteria}
    \sup_{u \in \mathcal{U}}~\left[\dot{h}(x,u) \triangleq \frac{\partial h}{\partial x}\left( f(x) + g(x)u\right)\right] \geq -\alpha(h(x)).
\end{equation}
\end{definition}

This ends our brief overview of necessary topics.  The next section motivates the specific problem under study.

\subsection{A Motivating Example}
\label{sec:motiv_example}
To better motivate our problem statement, we will provide a brief example.  Consider the following single integrator system subject to an STL specification $\psi$ with associated robustness measure $\rho$ and with $g = [0.75,0.75]^T$:
\begin{gather}
    \label{eq:exsingleint}
    \dot x = u,~x \in [-1,1]^2,~u \in [-0.5,0.5]^2, \\
    \mu_g(x) = \true \iff \left(h_{\mu}(x) \triangleq 0.1 - \|x-g\|_2\right) \geq 0, \\
    \psi = \F_{[0,2]}\mu_g,~
    \rho(s,0) \triangleq \max_{t \in [0,2]}~h_{\mu}\left(s(t)\right).
    \label{eq:example_spec}
\end{gather}
It is fairly simple to construct a controller $U$ that ensures that $(\phi(\mathbf{0}),0)\models \psi$, where $\phi(\mathbf{0}) \in \signalspace$ is the closed-loop solution of~\eqref{eq:exsingleint} and this controller $U$ starting from $x_0 = \mathbf{0}$.  Figure~\ref{fig:example_trajectory} shows an example controller and resulting trajectory $\phi(\mathbf{0})$.  Indeed, this controller also ensures that $\rho(\phi(\mathbf{0}),0) = 0.09$, indicating that this controller robustly steers the system to satisfy $\psi$.  However, if we introduce some disturbance to the system, as shown via the red trajectory in the same figure, the system fails to satisfy $\psi$.  As a result, the controller is not as robust as once claimed.  It is for this reason that we aim to develop techniques to discern the level of robustness - in a two-norm sense - that a controller can reject while still ensuring STL specification satisfaction.  Such techniques would provide a better understanding of the efficacy of a controller in robustly realizing a required task.  With this motivation in mind, we will formalize our problem statement.

\subsection{Problem Statement}
\label{sec:probstate}
We will start by mentioning two, separate systems - our nominal controlled system and its perturbed version.
\begin{align}
        \dot x & = f(x) + g(x)U(x)\triangleq f_{cl}(x), & & \hspace{-0.075 in}x \in \mathcal{X},~U:\mathcal{X} \to \mathcal{U},\quad~~~\label{eq:CL}\tag{CL} \\
        \dot x & = f_{cl}(x) + d, & &\hspace{-0.075 in}d \in \mathbb{R}^n. \label{eq:perturbed} \tag{CL-d}
\end{align}
For both closed-loop systems~\eqref{eq:CL} and~\eqref{eq:perturbed}, we will assume $f,g,U$ are locally Lipschitz continuous.  This implies $\forall~x_0 \in \mathcal{X}$ that solutions $\phi(x_0)$ to~\eqref{eq:CL} and $\phi^d(x_0)$ to~\eqref{eq:perturbed} have nonzero intervals of existence $I(x_0)$ and $I^d(x_0)$ respectively~\cite{verhulst2006nonlinear}.  Furthermore, we will denote $\phi(x_0) \in \signalspace$ to be the state trajectory signal and $\phi_t(x_0) \in \mathcal{X}$ to be the state at time $t$ as per equation~\eqref{eq:solution}.

We will also assume that this system is subject to an STL specification that is of the following form:
\begin{equation}
    \label{eq:probspec}
        \begin{gathered}
            \omega = \true~|~\mu~|~\neg \mu~|~\omega_1 \wedge \omega_2, \\
            \psi = \G_{[a,b]}\omega~|~\F_{[a,b]}\omega~|~\omega_1\until_{[a,b]}\omega_2~|~\psi_1\wedge\psi_2.
        \end{gathered}
\end{equation}
Additionally, we will make the following two assumptions about the predicate functions $h_{\mu}$ and the robustness measures $\rho$ used in our forthcoming analysis.
\begin{assumption}
\label{assump:predicate_func}
The predicate functions $h_{\mu}$ are continuously differentiable.
\end{assumption}
\begin{assumption}
\label{assump:robustness_lip}
The robustness measures $\rho$ for our signal temporal logic specifications $\psi$ are partially Lipschitz continuous, \textit{i.e.} $\exists~L,b\geq 0$ such that,
\begin{equation}
    |\rho(s,0) - \rho(z,0)| \leq L \|s-z\|_{[0,b]},
\end{equation}
where $\|\cdot\|_{[0,b]}$ is an induced (semi)-norm over $\signalspace$.
\end{assumption}
Here, we note that our restriction to this specific subclass of STL specifications aligns with prior work coupling Signal Temporal Logic and control barrier functions (see the examples in ~\cite{lindemann2018control,lindemann2019control,lindemann2019decentralized,lindemann2020barrier}).  We will also make one fairness assumption - that the intervals of existence for solutions to either system~\eqref{eq:CL} or~\eqref{eq:perturbed} are sufficiently large enough to permit analysis as to whether they satisfy their STL specification.  We will also state one definition to formalize what we mean when we say a system satisfies a specification.
\begin{definition}
\label{def:satisfactionnotation}
We say~\eqref{eq:CL} satisfies a specification $\psi$ over the space $X$, \textit{i.e.} \eqref{eq:CL} $\models_X \psi$ if and only if,
\begin{equation}
    \forall~x \in X,~ (\phi(x),0) \models \psi.
\end{equation}
\end{definition}
Then our problem statement is as follows.
\begin{problem}
Let $\psi$ be a Signal Temporal Logic specification of the form in~\eqref{eq:probspec}.  Determine a space $X \subseteq \mathcal{X}$ and a disturbance bound $\delta_d$ such that $\eqref{eq:perturbed}\models_X \psi~\forall~d \suchthat \|d\|\leq \delta_d$.
\end{problem}

In the sequel, the following definition of $P(\omega)$ for specifications $\omega$ as defined in equation~\eqref{eq:probspec} will be useful:
\begin{equation}
    \label{eq:truthset}
    \begin{gathered}
    \mu \in \omega \iff \left( \omega(x) = \true \implies \mu(x) = \true\right), \\
    P(\omega) = \{\mu \in \mathcal{A}~|~\mu \in \omega\}.
    \end{gathered}
\end{equation}
This results in the following Lemma.
\begin{lemma}
\label{lem:equivalency}
The following equivalency holds:
\begin{equation}
    \omega(x) \equiv \left( \wedge_{\mu \in P(\omega)}~\mu(x)\right).
\end{equation}
\end{lemma}
\begin{proof}
Follows by definition of $\omega$~\eqref{eq:probspec} and  $P(\omega)$~\eqref{eq:truthset}.
\end{proof}

This ends our formal problem statement.  We will now move to detailing our main contributions.

\section{Main Contribution}
\label{sec:contributions}
This section will be a series of optimization problems designed to identify spaces $X$ and norm bounds $\delta_d$ such that $\eqref{eq:perturbed}\models_X\psi$ for any STL specification $\psi$ satisfying equation~\eqref{eq:probspec}.  We will start first with an optimization problem for specifications $\psi = \G_{[0,b]}\omega$.  We do so as these types of specifications admit a time-independent solution worth noting.  As a result, our setting for the first optimization problem is as follows, with sets $\mathcal{C}_{h_{\mu}}$ as per equation~\eqref{eq:zerosuplevel}:
\begin{equation}
    \label{eq:first_setting}
        \psi = \G_{[0,b]}\omega,\quad \mathcal{C}_{\omega} = \mathcal{X} \bigcap_{\mu \in P(\omega)}~\mathcal{C}_{h_{\mu}}.
\end{equation}
We will also define a feasible disturbance set $\Delta$ as follows:
\begin{equation}
    \begin{gathered}
        \xi(x,e,\mu) = \frac{\partial h_{\mu}}{\partial x}^T(x)f_{cl}(x) - \left\|\frac{\partial h_{\mu}}{\partial x}(x)\right\|e,\\
        \Delta(x,\mu,\alpha_{\mu}) = \left\{e \in \mathbb{R}~|~\xi(x,e,\mu) \geq -\alpha_{\mu}\left(h_{\mu}(x)\right)\right\}.
    \end{gathered}
\end{equation}
Then our proposed optimization problem determines an $\omega$-specific  bound $\delta^0_d$ over $\mathcal{C}_{\omega}$ such that $\eqref{eq:perturbed} \models_{\mathcal{C}_{\omega}}\psi$ \textit{i.e.},
\begin{equation}
    \label{eq:candidate_bound_Gbase}
    \begin{aligned}
        \delta^0_d =&  \min_{x \in \mathcal{C}_{\omega}}~\max_{e \in \mathbb{R}} & & e, \\
        & \subjectto & & e \in \Delta(x,\mu,\alpha_{\mu}),~\forall~\mu \in P(\omega).
    \end{aligned}
\end{equation}
The formal statement of this theorem will follow.
\begin{theorem}
\label{thm:global_0on}
For equation~\eqref{eq:candidate_bound_Gbase}, let each $\alpha_{\mu} \in \classkapextinf$, let the specification $\psi$ and set $\mathcal{C}_{\omega}$ satisfy equation~\eqref{eq:first_setting}, and let each predicate function $h_{\mu}$ satisfy Assumption~\ref{assump:predicate_func}. Then,
\begin{equation}
    \delta^0_d \geq 0 \implies \eqref{eq:perturbed} \models_{\mathcal{C}_{\omega}} \psi~\forall~d \suchthat \|d\|\leq \delta^0_d.
\end{equation}
\end{theorem}
\begin{proof}
To start, for any $d$, Cauchy-Schwarz provides that
\begin{equation}
    \frac{\partial h_{\mu}}{\partial x}^T(x)\left(f_{cl}(x) +d \right) \geq 
    \frac{\partial h_{\mu}}{\partial x}^T(x)f_{cl}(x) - \left\|\frac{\partial h_{\mu}}{\partial x}(x)\right\|\|d\|.
\end{equation}
Then for any $d$ such that $\|d\|\leq \delta^0_d$ we have that the derivative of $h_{\mu}$ with respect to the perturbed dynamics~\eqref{eq:perturbed} satisfies the following inequality reminiscent of the CBF inequality in Definition~\ref{def:continuous_cbf} as $\delta^0_d \geq 0$:
\begin{equation}
    \dot{h}_{\mu}(x,d) \geq -\alpha_{\mu}(h_{\mu}(x)),~\forall~\mu \in P(\omega),~x \in \mathcal{C}_{\omega}.
\end{equation}
Via Peano's Uniqueness Theorem (Theorem 1.3.1 in~\cite{agarwal1993uniqueness}) we know that $\dot u = -\alpha_{\mu}(u)$ has a unique solution $\forall~u_0 \geq 0$ as $-\alpha_{\mu}$ is a continuous, non-increasing function in $u$.  Using this uniqueness result in conjunction with a Comparison Lemma, Lemma 3.4 in~\cite{Khalil}, allows us to state that
\begin{equation}
    \label{eq:continued_positivity}
    h_{\mu}\left(\phi^d_t(x_0)\right) \geq 0,~\forall~\mu \in P(\omega),~x_0 \in \mathcal{C}_{\omega},~t \in I^d(x_0).
\end{equation}
Here, we note that this chain of logic was also utilized in the proof for Theorem 1 in~\cite{ames2016control} as the proof for Lemma~3.4 in~\cite{Khalil} requires Lipschitz continuity of $\alpha_\mu$ to guarantee a unique solution (see Appendix C.2 in~\cite{Khalil}), and this is already provided for via Peano's Uniqueness Theorem.  As a result, equation~\eqref{eq:continued_positivity} implies that
\begin{equation}
    h_{\mu}\left(\phi^d_t(x_0)\right) \geq 0,~\forall~\mu \in P(\omega),~x_0 \in \mathcal{C}_{\omega},~t \geq 0.
\end{equation}
By definition of $h_{\mu}$ we have that
\begin{equation}
    \eqref{eq:perturbed} \models_{\mathcal{C}_{\omega}}\G_{[0,b]}\left(\wedge_{\mu \in P(\omega)}\mu\right)~\forall~d \suchthat \|d\|\leq \delta^0_d.
\end{equation}
Then by Lemma~\ref{lem:equivalency} and equation~\eqref{eq:first_setting} we have the following:
\begin{equation}
    \eqref{eq:perturbed} \models_{\mathcal{C}_{\omega}}\psi~\forall~d \suchthat \|d\|\leq\delta^0_d. \vspace{-0.625 cm}
\end{equation}
\spacing \spacing
\end{proof}

While this result may seem similar to work regarding Input to State Safe control barrier functions~\cite{kolathaya2018input}, such work tends to enlarge the safe-set to account for disturbances.  As our safe-set - \textit{e.g.} the region where $\mu$ is true - is fixed, we require an analysis that does not enlarge the safe set while still accounting for disturbances, resulting in our Theorem~\ref{thm:global_0on}.

\begin{figure}
    \centering
    \hspace{-0.15 cm}\includegraphics[width = 0.49\textwidth]{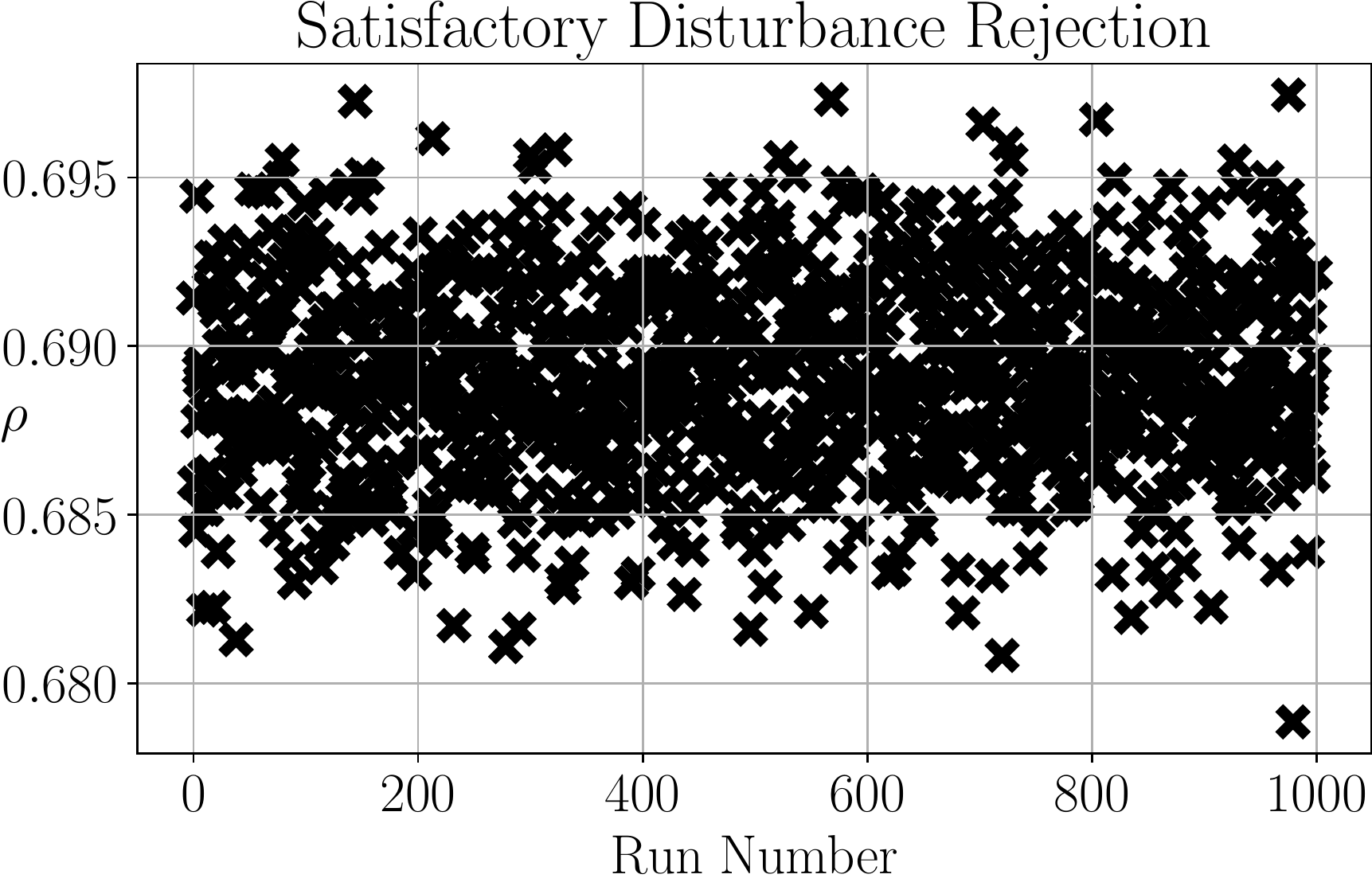}
    \vspace{-0.5 cm}
    \caption{The robustness measure for 1000 trials of the Segway detailed in Section~\ref{sec:examples} when perturbed by randomly distributed disturbances whose two-norm is less than the upper bound calculated by Theorem~\ref{thm:global_0on}, $\delta^0_d = 0.89$.  The robustness measure $\rho$ is for the specification $\psi_2 = \G_{[0,2]}\mu_2$ as per equation~\eqref{eq:example_spec}.  In all cases, the system satisfies its specification as $\rho\left(\phi^d(x_0,0)\right)\geq 0$.  This success indicates that, with high probability, this Segway's LQR controller rejects disturbances whose norm $\|d\| \leq \delta^0_d$.}
    \vspace{-0.5 cm}
    \label{fig:global_repeatability}
\end{figure}

For the second set of optimization problems, we will require the Gronwall-Bellman Inequality.
\begin{theorem}[From Theorem 1.3.1 in~\cite{ames1997inequalities}]
\label{thm:gronwall_bellman}
Let $u,f:J = [\alpha,\beta] \to \mathbb{R}_+$ be continuous over their domain, and let $n: J\to \mathbb{R}_+$ be continuous and non-decreasing.  Then, $\forall~t \in J$
\begin{equation}
    \begin{gathered}
    u(t) \leq n(t) + \int_{\alpha}^t f(x) u(s) ds \implies \\ u(t) \leq n(t) \exp\left(\int_{\alpha}^t f(s) ds \right).
    \end{gathered}
\end{equation}
\end{theorem}
This theorem allows us to establish the following lemma bounding the difference between solutions to dynamical systems~\eqref{eq:CL} and~\eqref{eq:perturbed}.
\begin{lemma}
\label{lem:solution_continuity}
For both systems~\eqref{eq:CL} and~\eqref{eq:perturbed}, let $f_{cl}$ be locally Lipschitz continuous with constant $L$ for some $x_0 \in \mathcal{X}$.  Then, if $\forall~d,~\|d\|\leq \delta_d$,
\begin{equation}
    \left\|\phi_t(x_0) - \phi^d_t(x_0)\right\| \leq \delta_dt e^{Lt},~\forall~t \in I(x_0) \cap I^d(x_0).
\end{equation}
\end{lemma}
\begin{proof}
This proof amounts to one application of Gronwall-Bellman's Inequality in Theorem~\ref{thm:gronwall_bellman}.  We can start with the norm difference between solutions which yields the following inequality for some $t \in I(x_0) \cap I^d(x_0)$:
\begin{equation}
    \begin{aligned}
    & \left\|\phi_t(x_0) - \phi^d_t(x_0)\right\| \leq \\
    & \quad \int_{0}^t \left\|f_{cl}\left(\phi_s(x_0)\right) - f_{cl}\left(\phi^d_s(x_0)\right)\right\| ds + \int_{0}^t \|d(s)\|ds.
    \end{aligned}
\end{equation}
By assumption that $f_{cl}$ is locally Lipschitz with constant $L$ and that all $d$ are such that $\|d\|\leq \delta_d$ we have that
\begin{equation}
\underbrace{\left\|\phi_t(x_0) - \phi^d_t(x_0)\right\|}_{u(t)} \leq \underbrace{\delta_d t}_{n(t)} + \int_{0}^t L \underbrace{\left\|\phi_s(x_0) - \phi^d_s(x_0)\right\|}_{u(s)} ds.
\end{equation}
Applying Theorem~\ref{thm:gronwall_bellman} concludes the proof.
\end{proof}

\begin{figure*}[ht]
    \centering
    \hspace{-0.1 cm}\includegraphics[width = \textwidth]{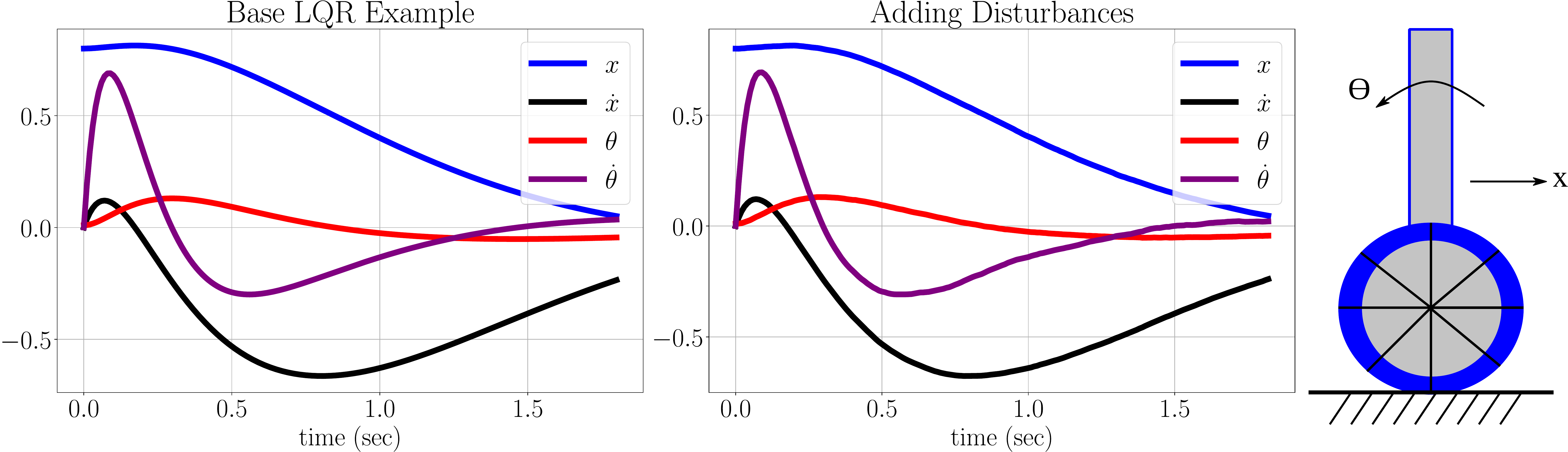}
    \vspace{-0.6 cm}
    \caption{Comparison of a Segway's LQR controller steering the nominal system~\eqref{eq:CL} to the zero point (left) and the disturbed system~\eqref{eq:perturbed} to the same zero point (center).  The example Segway is illustrated to the far right.  Notice that when the Segway undergoes disturbances whose norms are less than the max bound calculated via our procedure $\delta^T_d(\psi) = 0.01$, the specification $\psi$ in equation~\eqref{eq:sim_setup} is still satisfied.}
    \vspace{-0.4 cm}
    \label{fig:sim_results}
\end{figure*}

Our optimization problem for the remainder of the base specification types $\G_{[a,b]}\omega, \F_{[a,b]}\omega, \omega_1 \until_{[a,b]}\omega_2$ will make use of Lemma~\ref{lem:solution_continuity} and Assumption~\ref{assump:robustness_lip} to generate disturbance-bounds $\delta^1_d$ for the entire state space $\mathcal{X}$.  More aptly, our setting is as follows, with $"|"$ demarcating different specifications:
\begin{gather}
    \label{eq:second_setting}
    \psi = \G_{[a,b]}\omega~|~\F_{[a,b]}\omega~|~\omega_1 \until_{[a,b]}\omega_2, \\
    \rho(s,0) \geq 0 \iff (s,0) \models \psi,~
    \Delta_d = \min_{x \in \mathcal{X}}~\rho\left(\phi(x), 0\right). \label{eq:Delta_d}
\end{gather}
Then our theorem identifying a disturbance-bound $\delta_d$ for specifications $\psi$ of the type in equation~\eqref{eq:second_setting} is as follows.
\begin{theorem}
\label{thm:other_specifications}
Let the closed-loop dynamics $f_{cl}$ be locally Lipschitz continuous with constant $L_f~\forall~x_0 \in \mathcal{X}$, let the specification $\psi$ be as per equation~\eqref{eq:second_setting}, and let the robustness measure $\rho$ also satisfy Assumption~\ref{assump:robustness_lip} with Lipschitz constant $L_{\rho}$ and time constant $b$.  If $\Delta_d \geq 0$,
\begin{equation}
    \label{eq:second_cand_bound}
    \eqref{eq:perturbed} \models_{\mathcal{X}}\psi~\forall~d \suchthat \|d\| \leq \frac{\Delta_d}{L_{\rho}b e^{L_{f}b}} \triangleq \delta^1_d.
\end{equation}
\end{theorem}
\begin{proof}
For this proof, we will assume that our disturbances $d$ are such that $\|d\|\leq M$, and show $M=\delta^1_d$.  As a result, by local Lipschitz continuity of $f_{cl}$ and Lemma~\ref{lem:solution_continuity} we have that $\forall~x_0 \in \mathcal{X}$,
\begin{equation}
    \left\|\phi_t(x_0) - \phi^d_t(x_0)\right\| \leq Mt e^{L_ft},~\forall~t \in I(x_0)\cap I^d(x_0).
\end{equation}
Then as the robustness measure $\rho$ satisfies Assumption~\ref{assump:robustness_lip} with Lipschitz constant $L_{\rho}$ and time constant $b$, we have that $\forall~x_0 \in \mathcal{X}$ and with $\|\cdot\|_{[0,b]}$ the induced signal norm,
\begin{equation}
    \left|\rho\left(\phi(x_0),0\right) - \rho\left(\phi^d(x_0),0\right) \right| \leq L_{\rho} \left\|\phi(x_0) - \phi^d(x_0)\right\|_{[0,b]}.
\end{equation}
Then, by definition of $\|\cdot\|_{[0,b]}$ and our fairness assumption that $b \in I(x_0) \cap I^d(x_0)~\forall~x_0 \in \mathcal{X}$, we have that
\begin{equation}
    L_{\rho} \left\|\phi(x_0) - \phi^d(x_0)\right\|_{[0,b]} \leq L_{\rho}M b e^{L_f b},~\forall~x_0 \in \mathcal{X}.
\end{equation}
As a result, with $M = \Delta_d/(L_{\rho}b e^{L_{f}b})$ we have that
\begin{equation}
\left|\rho\left(\phi(x_0),0\right) - \rho\left(\phi^d(x_0),0\right) \right| \leq \Delta_d,~\forall~x_0 \in \mathcal{X}.
\end{equation}
By definition of $\Delta_d$ and $M$ and the above inequality holding $\forall~x_0 \in \mathcal{X}$, we have that
\begin{equation}
    \rho\left(\phi^d(x_0),0\right) \geq 0,~\forall~x_0,d \suchthat x_0 \in \mathcal{X},~\|d\| \leq \frac{\Delta_d}{L_{\rho}b e^{L_{f}b}}.
\end{equation}
Then the result follows by Definitions~\ref{def:robustness} and~\ref{def:satisfactionnotation}.
\end{proof}

Now it remains to identify a composite disturbance-bound for specifications $\psi = \wedge_i \psi_i$ where each $\psi_i$ is one of the base specification forms already accounted for, \textit{i.e.} $\G_{[0,b]}\omega,\G_{[a,b]}\omega,\F_{[a,b]}\omega,$ or $\omega_1 \until_{[a,b]}\omega_2$.  To do so, we will define an inclusion symbol for specifications.
\begin{gather}
    \psi_i \in \psi \iff \psi = \wedge_i \psi_i,~ P^1(\psi) = \{\psi'~|~\psi' \in \psi\}, \label{eq:spectruthset} \\
    \mathrm{\textit{e.g.}~for}~\psi = \psi_1 \wedge (\psi_2 \wedge \psi_3),~\psi_1,\psi_2,\psi_3 \in \psi.
\end{gather}
This leads to the following lemma similar to Lemma~\ref{lem:equivalency}.
\begin{lemma}
\label{lem:specequivalent}
The following statement holds.
\begin{equation}
    (s,0) \models \psi \iff (s,0) \models \psi',~\forall~\psi' \in P^1(\psi).
\end{equation}
\end{lemma}
\begin{proof}
We first note that the satisfaction operator $\models$ is recursively defined for the conjunction operator as follows~\cite{maler2004monitoring}:
\begin{equation}
    \label{eq:recursion_satisfaction}
    (s,0) \models \psi_1 \wedge \psi_2 \iff (s,0) \models \psi_1 \wedge (s,0) \models \psi_2,
\end{equation}
Then the proof follows by equation~\eqref{eq:recursion_satisfaction}, the definition of $P^1(\psi)$ in~\eqref{eq:spectruthset}, the types of specifications $\psi$ as per equation~\eqref{eq:probspec}, and the associativity of the $\wedge$ (and) operator.
\end{proof}

Then our final theorem determines a disturbance-bound $\delta_d$ for specifications $\psi = \wedge_i \psi_i$ where each $\psi_i$ is one of the base specification forms mentioned prior.  We will first pose our optimization problem, then state our theorem.
\begin{gather}
    \delta(\psi) \triangleq \begin{cases}
    \delta^0_d~\mathrm{as~per~}\eqref{eq:candidate_bound_Gbase} & \mbox{if~} \psi~\mathrm{is~as~per~}\eqref{eq:first_setting}, \\
    \delta^1_d~\mathrm{as~per~}\eqref{eq:second_cand_bound} & \mbox{else},
    \end{cases} \\
    \delta^T_d(\psi) = \min_{\psi_i \in P^1(\psi)}~\delta(\psi_i), \label{eq:overalldistbound} \\
    \mathcal{C}_{\psi} = \mathcal{X} \bigcap_{\substack{\psi_i \in P^1(\psi) \suchthat \\ \psi_i~\mathrm{as~per~}\eqref{eq:first_setting}}}\mathcal{C}_{\omega}~\mathrm{as~per~}\eqref{eq:first_setting}. \label{eq:final_set}
\end{gather}
\begin{theorem}
\label{thm:composite}
Let the system's specification $\psi$ satisfy~\eqref{eq:probspec} and let the assumptions for Theorems~\ref{thm:global_0on} and~\ref{thm:other_specifications} hold.  If $\delta^T_d(\psi) \geq 0$ with $\delta^T_d(\psi)$ as per equation~\eqref{eq:overalldistbound}, then
\begin{equation}
    \eqref{eq:perturbed} \models_{\mathcal{C}_{\psi}} \psi~\forall~d \suchthat \|d\|\leq \delta_d^T(\psi). 
\end{equation}
\end{theorem}
\begin{proof}
To start, we can assume without loss of generality that there exist zero or more specifications $\psi_i \in P^1(\psi)$ that are of the form in equation~\eqref{eq:first_setting}.  By definition of $\delta^T_d(\psi)$ in equation~\eqref{eq:overalldistbound}, $\mathcal{C}_{\psi}$ in equation~\eqref{eq:final_set}, and Theorem~\ref{thm:global_0on}, we have for each such specification $\psi_i$ (should they exist),
\begin{equation}
    \eqref{eq:perturbed} \models_{\mathcal{C}_{\psi}} \psi_i~\forall~d \suchthat \|d\|\leq \delta^T_d(\psi).
\end{equation}
This follows as if we have two sets $A,B$ such that $A \subset B$, a system $S$, and a specification $\psi$, then by Definition~\ref{def:satisfactionnotation},
\begin{equation}
    S \models_{B} \psi \implies S \models_{A} \psi.
\end{equation}
Then we can also assume without loss of generality that we have zero or more specifications $\psi_j \in P^1(\psi)$ such that $\psi_j$ are not of the form in equation~\eqref{eq:first_setting}.  For each such $\psi_j$, by definition of $\delta^T_d(\psi)$, $\mathcal{C}_{\psi}$, and Theorem~\ref{thm:other_specifications}, we have that
\begin{equation}
    \eqref{eq:perturbed} \models_{\mathcal{C}_{\psi}} \psi_j~\forall~d \suchthat \|d\|\leq \delta^T_d(\psi).
\end{equation}
Then the result holds via Lemma~\ref{lem:specequivalent}.
\end{proof}

This ends the series of optimization problems to determine our disturbance-bounds.  We will now move to showcase these results through a simulated example on a Segway.
\section{Simulated Examples}
\label{sec:examples}
For our example, we aim to determine the robustness with which a Segway's LQR controller achieves two desired performance bounds.  First, the Segway's pendulumn angle is never to deviate too far from the vertical.  Second, the Segway is to reach its goal - its state $\mathbf{x}$ should lie within a norm bounded ball around $0$ - within two seconds.  Mathematically this leads to the following setting:
\begin{gather}
    h_1(\mathbf{x}) = 0.25 - \|\mathbf{x}\|,~h_2(\mathbf{x}) = 10(0.3^2 - \theta^2) - 2\theta\dot\theta,~~ \label{eq:sim_setup}\\
    \mu_i(x) \equiv \left(h_i(x) \geq 0\right),~
    \psi = \F_{[0,2]} \mu_1 \wedge \G_{[0,2]}\mu_2, \label{eq:ex_spec} \\
    \mathcal{X} \subset [-1,1]^2 \times [-0.4,0.4] \times [-1.5,1.5], \\
    \mathbf{x} = [x,v,\theta,\dot \theta]^T \in \mathcal{X} \subset \mathbb{R}^4.
\end{gather}
Figure~\ref{fig:sim_results} shows the Segway setup and example LQR controller steering the Segway to satisfy this specification $\psi$.

To start, it is clear that both predicate functions $h_1,h_2$ in equation~\eqref{eq:sim_setup} satisfy Assumption~\ref{assump:predicate_func}.  Indeed as both are Lipschitz continuous, so to are the associated robustness measures generated from these predicate functions Lipschitz continuous as well, which satisfies Assumption~\ref{assump:robustness_lip}.  As a result, we break our specification into two parts as required of Theorem~\ref{thm:composite} - $\psi_1 = \F_{[0,2]}\mu_1$ and $\psi_2 = \G_{[0,2]}\mu_2$.  This resulted in a $\delta^0_d = 0.89$ after utilizing Theorem~\ref{thm:global_0on} for $\psi_2$ and a $\Delta_d = 0.2$ after utilizing Theorem~\ref{thm:other_specifications} for $\psi_1$.

Figure~\ref{fig:global_repeatability} shows the results of $1000$ randomized trials of the Segway undergoing disturbances $d$ such that $\|d\|\leq\delta^0_d = 0.89$.  As can be seen, the LQR controller realizes a positive robustness measure indicating that the system-controller pair can reject disturbances whose norm is under the bound we identify through our procedure.  Additionally, under the assumption that our Segway's closed-loop dynamics $f_{cl}$ are Lipschitz continuous with constant $L_f \leq 1$ and knowing the associated robustness measure $\rho$ for $\mu_1$ as per~\eqref{eq:ex_spec} is Lipschitz continuous with $L_{\rho} = 1$, Theorem~\ref{thm:other_specifications} provides a secondary disturbance-bound $\delta^1_d = 0.01$.  As per Theorem~\ref{thm:composite} this indicates that our Segway should satisfy its overall specification $\psi$ if its disturbance $d$ is such that $\|d\| \leq \delta^T_d(\psi) = 0.01$.  Indeed the Segway does satisfy its specification after $1000$ randomized runs when perturbed by normally distributed disturbances $d$ such that $\|d\| \leq 0.01$.  One such run is shown in Figure~\ref{fig:sim_results}.

\section{Conclusion}
In this paper, we constructed a series of optimization problems to determine the level of disturbance - in a two-norm sense - that a given system's controller can reject while satisfying its operational Signal Temporal Logic specification.  Additionally, we showed that our optimization problems generate reasonable disturbance-bounds through simulating a Segway whose dynamics are perturbed by disturbances whose two-norm is less than our calculated bound.  Future work aims to decrease the conservativeness of our calculated bounds and extend the class of specifications capable of being analyzed by our approach.

\bibliographystyle{IEEEtran}
\bibliography{IEEEabrv,bib_works}

\end{document}